\documentclass[a4paper,12pt]{article}

\usepackage{color}
\usepackage{amssymb, amsmath}
\usepackage{hyperref}
\usepackage{tikz}

\newcommand\EFFACE[1]{}

\newtheorem{theorem}{Theorem}
\newtheorem{proposition}[theorem]{Proposition}

\newtheorem{corollary}[theorem]{Corollary}

\newtheorem{observation}[theorem]{Observation}

\newenvironment{proof}{
\par
\noindent {\bf Proof.}\rm}{\mbox{}\hfill$\square$\par\vskip 3mm}

\def\ZZZ{\mathbb{Z}}

\makeatletter
\let\@fnsymbol\@arabic
\makeatother

\begin{document}


\title{{\bf 2-Distance Colorings of Integer Distance Graphs}} 


\author{Brahim BENMEDJDOUB~\thanks{Faculty of Mathematics, Laboratory L'IFORCE, University of Sciences and Technology
Houari Boumediene (USTHB), B.P.~32 El-Alia, Bab-Ezzouar, 16111 Algiers, Algeria.}
\and \'Eric SOPENA~\thanks{Univ. Bordeaux, LaBRI, UMR5800, F-33400 Talence, France.}~$^,$\thanks{CNRS, LaBRI, UMR5800, F-33400 Talence, France.}~$^,$\footnote{Corresponding author. Eric.Sopena@labri.fr.}
\and Isma BOUCHEMAKH~\footnotemark[1]
}

\maketitle

\abstract{
A 2-distance $k$-coloring of a graph $G$ is a mapping from $V(G)$ to the set of colors $\{1,\dots,k\}$
such that every two vertices at distance at most 2 receive distinct colors.
The 2-distance chromatic number $\chi_2(G)$ of $G$ is then the mallest $k$ for which $G$
admits a 2-distance $k$-coloring.
For any finite set of positive integers $D=\{d_1,\dots,d_k\}$, 
the integer distance graph $G=G(D)$ is the infinite graph
defined by $V(G)=\ZZZ$ and $uv\in E(G)$ if and only if $|v-u|\in D$.
We study the 2-distance chromatic number of integer distance graphs for several types of sets $D$.
In each case, we provide exact values or upper bounds on this parameter 
and characterize those graphs $G(D)$ with $\chi_2(G(D))=\Delta(G(D))+1$.
}

\medskip



\section{Introduction}

All the graphs we consider in this paper are simple and loopless undirected graphs. 
We denote by $V(G)$ and $E(G)$ the set of vertices and the set of edges of a graph $G$, respectively.
For any two vertices $u$ and $v$ of $G$, we denote by $d_G(u,v)$ the \emph{distance} between $u$ and $v$,
that is the length of a shortest path joining $u$ and $v$.
We denote by $\Delta(G)$ the maximum degree of $G$.

A (proper) \emph{$k$-coloring} of a graph $G$ is a mapping from $V(G)$ to the set of colors $\{1,\dots,k\}$
such that every two adjacent vertices receive distinct colors.
The smallest $k$ for which $G$
admits a $k$-coloring is the \emph{chromatic number} of $G$, denoted $\chi(G)$.
A \emph{2-distance $k$-coloring} of a graph $G$ is a mapping from $V(G)$ to the set of colors $\{1,\dots,k\}$
such that every two vertices at distance at most 2 receive distinct colors. 
2-distance colorings are sometimes called \emph{L(1,1)-labelings} (see~\cite{Ca11} for a survey
on $L(h,k)$-labelings) or \emph{square colorings} in the literature.
The smallest $k$ for which $G$
admits a 2-distance $k$-coloring is the \emph{2-distance chromatic number} of $G$, denoted $\chi_2(G)$.

The \emph{square} $G^2$ of a graph $G$ is the graph defined by $V(G^2)=V(G)$ and $uv\in E(G^2)$
if and only if $d_G(u,v)\le 2$.
Clearly, a 2-distance coloring of a graph $G$ is nothing but a proper coloring of $G^2$ and, therefore,
$\chi_2(G)=\chi(G^2)$ for every graph~$G$.

The study of 2-distance colorings was initiated by Kramer and Kramer~\cite{KK69} (see also their survey 
on general distance colorings in~\cite{KK08}).
The case of planar graphs has attracted a lot of attention
in the literature (see e.g. \cite{AH03,BLP14a,BLP14b,BI09,DKNS08,LW06,WL03}),
due to the conjecture of Wegner that suggests an upper bound
on the 2-distance chromatic number of planar graphs depending on their maximum
degree (see~\cite{W77} for more details).

\medskip

In this paper, we study 2-distance 
colorings of 
distance graphs.
Although several coloring problems have been considered for distance graphs (see~\cite{Li08} for a survey),
it seems that 2-distance 
colorings have not been considered yet.
We present in Section~\ref{sec:preliminaries} a few basic results on the chromatic number of distance graphs.
We then consider specific sets $D$, namely
$D=\{1,a\}$, $a\ge 3$ (in Section~\ref{sec:1a}),
$D=\{1,a,a+1\}$, $a\ge 3$ (in Section~\ref{sec:1aa+1}),
and $D=\{1,\ldots,m,a\}$, $2\le m<a$ (in Section~\ref{sec:1toma}).
We finally 
propose some open problems in Section~\ref{sec:discussion}.


\section{Preliminaries}
\label{sec:preliminaries}

Let $D=\{d_1,\dots,d_k\}$ be a finite set of positive integers. 
The \emph{integer distance graph} (simply called {\em distance graph} in the following) $G=G(D)$ is the infinite graph
defined by $V(G)=\ZZZ$ and $uv\in E(G)$ if and only if $|v-u|\in D$.

If $\gcd(\{d_1,\dots,d_k\})=p>1$, the distance graph $G(D)$ 
has $p$ connected components, each of them being isomorphic to
the distance graph $G(D')$ with $D'=\{d_1/p,\dots,d_k/p\}$.
In that case, we thus have $\chi_2(G(D))=\chi_2(G(D'))$
so that we can always assume $\gcd(D)=1$.

It is easy to observe that the square of the distance graph $G(D)$ is
also a distance graph, namely the distance graph $G(D^2)$ where
$$D^2=D\ \cup\ \{d+d'\ /\ d,d'\in D\}\ \cup\ \{d-d'\ /\ d,d'\in D,\ d>d'\}.$$
For instance, for $D=\{1,2,5\}$, we get $D^2=\{1,2,3,4,5,6,7,10\}$. Note that
if $D$ has cardinality $k$, then $D^2$ has cardinality at most $k(k+1)$.

As observed in the previous section, $\chi_2(G)=\chi(G^2)$ for every graph $G$.
Therefore, since $(G(D))^2=G(D^2)$, determining the 2-distance chromatic number of the distance graph
$G(D)$ reduces to determining the chromatic number of the
distance graph $G(D^2)$.
The problem of determining the chromatic number of distance graphs has been
extensively studied in the literature. 
When $|D|\le 2$, this question is easily solved, thanks to the following 
general upper bounds:

\begin{proposition}[folklore]
For every finite set of positive integers $D=\{d_1,\dots,d_k\}$
and every positive integer $p$ such that $d_i\not\equiv 0\pmod p$ for every $i$, $1\le i\le k$,
$\chi(G(D))\le p$.
\label{prop:k-upper-bound}
\end{proposition}

\begin{proof}
Let $\lambda:V(G(D))\longrightarrow\{1,\dots,p\}$ be the mapping
defined by 
$$\lambda(x)=1+(x\mod p),$$ 
for every integer $x\in\ZZZ$.
Since $d_i\not\equiv 0\pmod p$ for every $i$, $1\le i\le k$,
the mapping $\lambda$ is clearly a proper coloring of $G(D)$.
\end{proof}

\begin{theorem}[Walther~\cite{W90}]
For every finite set of positive integers $D$, 
$$\chi(G(D))\le|D|+1.$$
\label{th:Walther-D+1}
\end{theorem}

\begin{proof}
A $(|D|+1)$-coloring of $G(D)$ can easily be produced using the First-Fit greedy algorithm, starting from vertex 0,
from left to right and then from right to left, since every vertex has exactly $|D|$ neighbors on its left and 
$|D|$ neighbors on its right.
\end{proof}

Therefore, when $|D|\le 2$, $\chi(G(D))=2$ if $|D|=1$ or all elements in $D$ are odd (since $G(D)$ is then bipartite),
and $\chi(G(D))=3$ otherwise (since $G(D)$ then contains cycles of odd length).
The case $|D|=3$ has been settled by Zhu~\cite{Z02}.
Whenever $|D|\ge 4$, only partial results have been obtained, namely for sets $D$
having specific properties.

A coloring $\lambda$ of a distance graph $G(D)$ is \emph{$p$-periodic}, for some
integer $p\ge 1$, if $\lambda(x+p)=\lambda(x)$ for every $x\in\ZZZ$. 
Walther also proved the following:

\begin{theorem}[Walther~\cite{W90}]
For every finite set of positive integers $D$, if $\chi(G(D))\le k$ then $G(D)$
admits a $p$-periodic $k$-coloring for some $p$.
\label{th:Walther-periodic}
\end{theorem}

The \emph{pattern} of such a $p$-periodic coloring is defined as the sequence $\lambda(x)\dots\lambda(x+p-1)$.
In particular,
the coloring defined in the proof of Proposition~\ref{prop:k-upper-bound} was $p$-periodic
with pattern $12\dots p$.
In the following, we will describe such patterns using standard notation of Combinatorics on words.
For instance, the pattern $121212345$ will be denoted $(12)^3345$.

Finally, note that in any 2-distance coloring of a graph $G$, all vertices in the closed neighborhood
of any vertex must be assigned distinct colors. Therefore, we have the following:

\begin{observation}
For every graph $G$, $\chi_2(G)\ge\Delta(G)+1$.
\label{obs:lower-bound}
\end{observation}

In particular, this bound is attained by the distance graph $G(D)$ with
$D=\{1,\dots,k\}$, $k\ge 2$:

\begin{proposition}
For every $k\ge 2$, $\chi_2(G(\{1,\dots,k\}))=2k+1=\Delta(G(\{1,\dots,k\}))+1$.
\label{prop:chi2-1k}
\end{proposition}

\begin{proof}
It is easy to check that the mapping $\lambda$
given by 
$$\lambda(x)=1+(x\mod{2k+1})$$
 for every $x\in\ZZZ$ is a 2-distance $(2k+1)$-coloring of $G(\{1,\dots,k\})$.
Equality then follows from Observation~\ref{obs:lower-bound}.
\end{proof}

%
%


\section{The case $D=\{1,a\}$, $a\ge 3$}
\label{sec:1a}

We study in this section the 2-distance 
chromatic number of distance graphs $G(D)$
with $D=\{1,a\}$, $a\ge 3$ 
(note that the case $a=2$ is already solved by Proposition~\ref{prop:chi2-1k}).

When $D=\{1,a\}$, $a\ge 3$, 
we have $\Delta(G(D))=4$ and
$$D^2=\{1,2,a-1,a,a+1,2a\}.$$

The following theorem gives the 2-distance chromatic number of any such graph:

\begin{theorem}
For every integer $a\ge 3$,
$$\chi_2(G(\{1,a\}))=\left\{
\begin{array}{ll}
5 & \mbox{if  $a\equiv 2\pmod 5$, or $a\equiv 3\pmod 5$,}\\
6 & \mbox{otherwise.}
\end{array}
\right.
$$
\label{th:chi_2-dist-1a}
\end{theorem}

\begin{proof}
Since $\{1,a\}^2=\{1,2,a-1,a,a+1,2a\}$,
we get $d\not\equiv 0\pmod 5$ for every $d\in \{1,a\}^2$
whenever $a\equiv 2\pmod 5$ or $a\equiv 3\pmod 5$
 and thus, by Proposition~\ref{prop:k-upper-bound}
and Observation~\ref{obs:lower-bound},
$\chi_2(G(\{1,a\}))=5$.

Note  that for every $x\in\ZZZ$, the set of vertices 
$$C(x)=\{x-a,x-1,x,x+1,x+a\}$$ 
induces a clique in $G(\{1,a\}^2)$  (see Figure~\ref{fig:1a}).
We now claim that every 2-distance 5-coloring $\lambda$ of $G(\{1,a\})$ is necessarily 5-periodic,
that is $\lambda(x+5)=\lambda(x)$ for every $x\in\ZZZ$.
To show that, it suffices to prove that any five consecutive vertices $x,\dots,x+4$ must be assigned distinct colors.
Assume to the contrary that this is not the case and, without loss of generality, let $x=0$.
Since vertices 0, 1 and 2 necessarily get distinct colors, we only have two cases to consider:

\begin{enumerate}

\item $\lambda(0)=\lambda(3)=1$, $\lambda(1)=2$, $\lambda(2)=3$.\\
Since $C(1)$ induces a clique in $G(\{1,a\}^2)$ (depicted in bold in Figure~\ref{fig:1a}), we have 
$$\{\lambda(1-a),\lambda(1+a)\}=\{4,5\},$$
which implies 
$$\{\lambda(2-a),\lambda(2+a)\}=\{4,5\}.$$
(More precisely, $\lambda(2-a)=9-\lambda(1-a)$ and $\lambda(2+a)=9-\lambda(1+a)$).
This implies $\lambda(3-a)=\lambda(3+a)=2$, a contradiction since $d(3-a,3+a)=2$.

\item $\lambda(0)=\lambda(4)=1$, $\lambda(1)=2$, $\lambda(2)=3$, $\lambda(3)=4$.\\
As in the previous case we have 
$$\{\lambda(1-a),\lambda(1+a)\}=\{4,5\},$$
which implies 
$$\{\lambda(2-a),\lambda(2+a)\}=\{1,5\}.$$
We then get $\lambda(3-a)=\lambda(3+a)=2$, again a contradiction.

\end{enumerate}

\begin{figure}
\begin{center}
\begin{tikzpicture}[domain=0:17,x=1cm,y=1cm]
\node[draw] (A) at (0,0){{\bf 0}};
\node[draw] (B) at (2,0){{\bf 1}};
\node[draw] (C) at (4,0){{\bf 2}};
\node[draw] (D) at (6,0){3};
\node[draw] (E) at (8,0){4};

\node[draw] (Ap) at (0,2){a};
\node[draw] (Bp) at (2,2){{\bf 1+a}};
\node[draw] (Cp) at (4,2){2+a};
\node[draw] (Dp) at (6,2){3+a};
\node[draw] (Ep) at (8,2){4+a};

\node[draw] (Am) at (0,-2){-a};
\node[draw] (Bm) at (2,-2){{\bf 1-a}};
\node[draw] (Cm) at (4,-2){2-a};
\node[draw] (Dm) at (6,-2){3-a};
\node[draw] (Em) at (8,-2){4-a};

\node at (-2,0){$\dots$};
\node at (10,0){$\dots$};
\draw[thick] (A) to (-1,0);
\draw[thick] (Ap) to (-1,2);
\draw[thick] (Am) to (-1,-2);
\draw[thick] (E) to (9,0);
\draw[thick] (Ep) to (9,2);
\draw[thick] (Em) to (9,-2);

\draw[very thick] (A) to (B);
\draw[very thick] (B) to (C);
\draw[thick] (C) to (D);
\draw[thick] (D) to (E);
\draw[thick] (Ap) to (Bp);
\draw[thick] (Bp) to (Cp);
\draw[thick] (Cp) to (Dp);
\draw[thick] (Dp) to (Ep);
\draw[thick] (Am) to (Bm);
\draw[thick] (Bm) to (Cm);
\draw[thick] (Cm) to (Dm);
\draw[thick] (Dm) to (Em);

%
\draw[thick] (A) to (Ap);
\draw[thick] (A) to (Am);
\draw[very thick] (B) to (Bp);
\draw[very thick] (B) to (Bm);
\draw[thick] (C) to (Cp);
\draw[thick] (C) to (Cm);
\draw[thick] (D) to (Dp);
\draw[thick] (D) to (Dm);
\draw[thick] (E) to (Ep);
\draw[thick] (E) to (Em);

\end{tikzpicture}
\caption{The distance graph $G(\{1,a\})$, $a\ge 3$}
\label{fig:1a}

\end{center}
\end{figure}
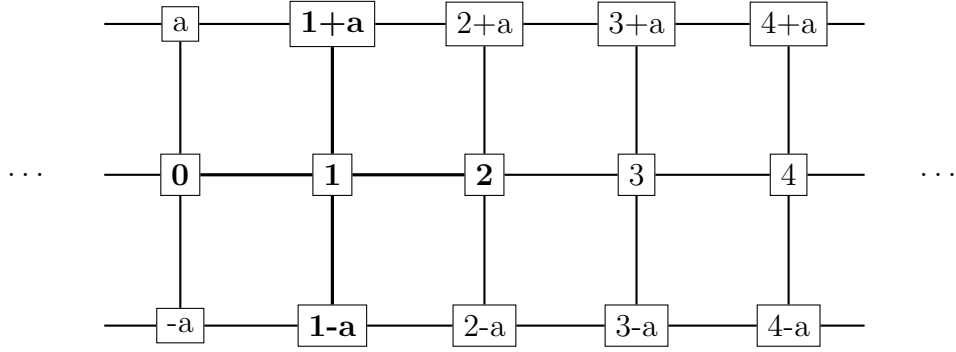

Therefore, $\chi_2(G(\{1,a\}))=5$ if and only if 5 do not divide any element of $\{1,a\}^2=\{1,2,a-1,a,a+1,2a\}$.
This is clearly the case if and only if $a\equiv 2\pmod 5$ or $a\equiv 3\pmod 5$. 

We finally prove that there exists a 2-distance 6-coloring of $G(\{1,a\})$ for any value of $a$.
We consider three cases, according to the value of $(a\mod 3)$:
\begin{enumerate}

\item $a=3k$, $k\ge 1$.\\
Let $\lambda$ be the $(2a-1)$-periodic mapping defined by the pattern
$$(123)^k(456)^{k-1}45.$$
If $\lambda(x)=\lambda(y)=c$, $1\le c\le 5$, then 
$$d(x,y)\in\{3q,\ 0\le q\le k-1\}\ \cup\ \{(2a-1)p+3q,\ p\ge 1,\ 1-k\le q\le k-1\}.$$
If $\lambda(x)=\lambda(y)=6$, then 
$$d(x,y)\in\{3q,\ 0\le q\le k-2\}\ \cup\ \{(2a-1)p+3q,\ p\ge 1,\ 2-k\le q\le k-2\}.$$
Therefore, in both cases, $d(x,y)\notin\{1,2,a-1,a,a+1,2a\}$, and thus $\lambda$ is a 2-distance 6-coloring of $G(\{1,a\})$.

\item $a=3k+1$, $k\ge 1$.\\
Let $\lambda$ be the $(2a-2)$-periodic mapping defined by the pattern
$$(123)^k(456)^{k}.$$
If $\lambda(x)=\lambda(y)=c$, $1\le c\le 6$, then 
$$d(x,y)\in\{3q,\ 0\le q\le k-1\}\ \cup\ \{(2a-2)p+3q,\ p\ge 1,\ 1-k\le q\le k-1\}.$$
Therefore, $d(x,y)\notin\{1,2,a-1,a,a+1,2a\}$, and thus $\lambda$ is a 2-distance 6-coloring of $G(\{1,a\})$.

\item $a=3k+2$, $k\ge 1$.\\
Let $\lambda$ be the $(2a+1)$-periodic mapping defined by the pattern
$$(123)^{k+1}(456)^{k}45.$$
If $\lambda(x)=\lambda(y)=c$, $1\le c\le 5$, then 
$$d(x,y)\in\{3q,\ 0\le q\le k\}\ \cup\ \{(2a+1)p+3q,\ p\ge 1,\ -k\le q\le k\}.$$
If $\lambda(x)=\lambda(y)=6$, then 
$$d(x,y)\in\{3q,\ 0\le q\le k-1\}\ \cup\ \{(2a+1)p+3q,\ p\ge 1,\ 1-k\le q\le k-1\}.$$
Therefore, in both cases, $d(x,y)\notin\{1,2,a-1,a,a+1,2a\}$, and thus $\lambda$ is a 2-distance 6-coloring of $G(\{1,a\})$.
\end{enumerate}
This concludes the proof.
\end{proof}


\section{The case $D=\{1,a,a+1\}$, $a\ge 3$}
\label{sec:1aa+1}

We study in this section the 2-distance 
chromatic number of distance graphs $G(D)$
with $D=\{1,a,a+1\}$, $a\ge 3$ 
(note that the case $a=2$ is already solved by Proposition~\ref{prop:chi2-1k}).


When $D=\{1,a,a+1\}$, $a\ge 3$, we have $\Delta(G(D))=6$ and
$$D^2=\{1,2,a-1,a,a+1,a+2,2a,2a+1,2a+2\}.$$

We first prove the following:

\begin{theorem}
For every integer $a$, $a\ge 3$, 
$$\chi_2(G(\{1,a,a+1\}))=7=\Delta(G(\{1,a,a+1\}))+1$$
 if and only if $a\equiv 2\pmod 7$ or $a\equiv 4\pmod 7$.
\label{th:chi_2-dist-1aa+1}
\end{theorem}

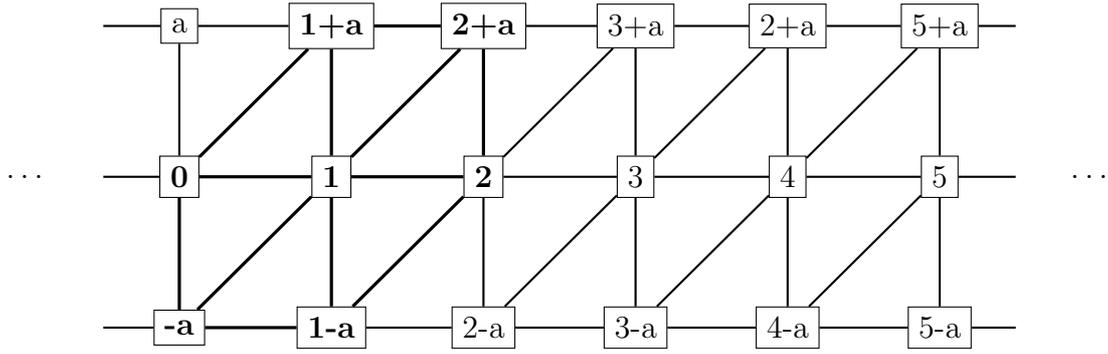
\begin{figure}
\begin{center}
\begin{tikzpicture}[domain=0:17,x=1cm,y=1cm]
\node[draw] (A) at (0,0){{\bf 0}};
\node[draw] (B) at (2,0){{\bf 1}};
\node[draw] (C) at (4,0){{\bf 2}};
\node[draw] (D) at (6,0){3};
\node[draw] (E) at (8,0){4};
\node[draw] (F) at (10,0){5};

\node[draw] (Ap) at (0,2){a};
\node[draw] (Bp) at (2,2){{\bf 1+a}};
\node[draw] (Cp) at (4,2){{\bf 2+a}};
\node[draw] (Dp) at (6,2){3+a};
\node[draw] (Ep) at (8,2){2+a};
\node[draw] (Fp) at (10,2){5+a};

\node[draw] (Am) at (0,-2){{\bf -a}};
\node[draw] (Bm) at (2,-2){{\bf 1-a}};
\node[draw] (Cm) at (4,-2){2-a};
\node[draw] (Dm) at (6,-2){3-a};
\node[draw] (Em) at (8,-2){4-a};
\node[draw] (Fm) at (10,-2){5-a};

\node at (-2,0){$\dots$};
\node at (12,0){$\dots$};
\draw[thick] (A) to (-1,0);
\draw[thick] (Ap) to (-1,2);
\draw[thick] (Am) to (-1,-2);
\draw[thick] (F) to (11,0);
\draw[thick] (Fp) to (11,2);
\draw[thick] (Fm) to (11,-2);

\draw[very thick] (A) to (B);
\draw[very thick] (B) to (C);
\draw[thick] (C) to (D);
\draw[thick] (D) to (E);
\draw[thick] (E) to (F);
\draw[thick] (Ap) to (Bp);
\draw[very thick] (Bp) to (Cp);
\draw[thick] (Cp) to (Dp);
\draw[thick] (Dp) to (Ep);
\draw[thick] (Ep) to (Fp);
\draw[very thick] (Am) to (Bm);
\draw[thick] (Bm) to (Cm);
\draw[thick] (Cm) to (Dm);
\draw[thick] (Dm) to (Em);
\draw[thick] (Em) to (Fm);

%
\draw[thick] (A) to (Ap);
\draw[very thick] (A) to (Am);
\draw[very thick] (B) to (Bp);
\draw[very thick] (B) to (Bm);
\draw[very thick] (C) to (Cp);
\draw[thick] (C) to (Cm);
\draw[thick] (D) to (Dp);
\draw[thick] (D) to (Dm);
\draw[thick] (E) to (Ep);
\draw[thick] (E) to (Em);
\draw[thick] (F) to (Fp);
\draw[thick] (F) to (Fm);

\draw[very thick] (A) to (Bp);
\draw[very thick] (B) to (Cp);
\draw[very thick] (B) to (Am);
\draw[thick] (C) to (Dp);
\draw[very thick] (C) to (Bm);
\draw[thick] (D) to (Ep);
\draw[thick] (D) to (Cm);
\draw[thick] (E) to (Fp);
\draw[thick] (E) to (Dm);
\draw[thick] (F) to (Em);

\end{tikzpicture}
\caption{The distance graph $G(\{1,a,a+1\})$, $a\ge 3$}
\label{fig:1aa+1}

\end{center}
\end{figure}

\begin{proof}
Since $\{1,a,a+1\}^2=\{1,2,a-1,a,a+1,a+2,2a,2a+1,2a+2\}$,
we get $d\not\equiv 0\pmod 7$ for every $d\in \{1,a,a+1\}^2$
whenever $a\equiv 2\pmod 7$ or $a\equiv 4\pmod 7$
 and thus, by Proposition~\ref{prop:k-upper-bound}
and Observation~\ref{obs:lower-bound},
$\chi_2(G(\{1,a,a+1\}))=7$.

Note  that for every $x\in\ZZZ$, the set of vertices 
$$C(x)=\{x-a-1,x-a,x-1,x,x+1,x+a,x+a+1\}$$ 
induces a clique in $G(\{1,a,a+1\}^2)$.
We now claim that every 2-distance $7$-coloring $\lambda$ of $G(\{1,a,a+1\})$ is necessarily $7$-periodic,
that is $\lambda(x+7)=\lambda(x)$ for every $x\in\ZZZ$.
To show that, it suffices to prove that any $7$ consecutive vertices $x,\dots,x+6$ must be assigned distinct colors.
Assume to the contrary that this is not the case and, without loss of generality, let $x=0$.
Since vertices $0$, $1$ and $2$ necessarily get distinct colors, we only have four cases to consider (see Figure~\ref{fig:1aa+1}):

\begin{enumerate}

\item Vertices $0,1,2,3$ are colored with the colors $1,2,3$ and 1, respectively.\\
We consider two subcases:
\begin{enumerate}
\item $\lambda(4)=2$.\\
Since $C(1)$ induces a clique in $G(\{1,a,a+1\}^2)$ (depicted in bold in Figure~\ref{fig:1aa+1}), we have 
$$\{\lambda(-a),\lambda(1-a),\lambda(1+a),\lambda(2+a)\}=\{4,5,6,7\}.$$
For similar reasons, we also have 
$$\{\lambda(2-a),\lambda(3-a),\lambda(3+a),\lambda(4+a)\}=\{4,5,6,7\}.$$
This implies $\lambda(-a)=\lambda(4-a)$ or $\lambda(1+a)=\lambda(5+a)$. Each of these cases thus corresponds
to case 2 below.

\item $\lambda(4)\notin\{1,2,3\}$.\\
Assume $\lambda(4)=4$, without loss of generality.
As in the previous subcase, we have
$$\{\lambda(-a),\lambda(1-a),\lambda(1+a),\lambda(2+a)\}=\{4,5,6,7\},$$
and, similarly,
$$\{\lambda(1-a),\lambda(2-a),\lambda(2+a),\lambda(3+a)\}=\{4,5,6,7\}.$$
Moreover, since $\lambda(4)=4$, we get
$$\{\lambda(3+a),\lambda(2-a)\}\subseteq\{5,6,7\}.$$
On the other hand, considering the clique 
$S(3)$ in $G(\{1,a,a+1\}^2)$, we also get
$$\{\lambda(4+a),\lambda(3-a)\}\subseteq\{5,6,7\}.$$
We thus get a contradiction since we only have three available colors
for the clique induced by the four vertices $2-a$, $3-a$, $a+3$ and $a+4$ in $G(\{1,a,a+1\}^2)$.
\end{enumerate}

\item Vertices $0,1,2,3,4$ are colored with the colors $1,2,3,4$ and 1, respectively.\\
Again considering cliques $C(2)$ and $C(3)$ in $G(\{1,a,a+1\}^2)$, we get
$$\{\lambda(1-a),\lambda(2+a)\}\subseteq\{5,6,7\},$$
and
$$\{\lambda(2-a),\lambda(3+a)\}\subseteq\{5,6,7\},$$
a contradiction since vertices $1-a$, $2-a$, $a+2$ and $a+3$ induce a clique in $G(\{1,a,a+1\}^2)$.

\item Vertices $0,1,2,3,4,5$ are colored with the colors $1,2,3,4,5$ and 1, respectively.\\
Considering the cliques $C(1)$, $C(2)$ and $C(3)$ in $G(\{1,a,a+1\}^2)$, we get
$$\{\lambda(-a),\lambda(1-a),\lambda(1+a),\lambda(2+a)\}=\{4,5,6,7\},$$
$$\{\lambda(2-a),\lambda(3+a)\} \subseteq \{1,\lambda(-a),\lambda(1+a)\}\setminus\{4,5\},$$
$$\{\lambda(3-a),\lambda(4+a)\} \subseteq \{2,\lambda(1-a),\lambda(2+a)\}\setminus\{4,5\},$$
and thus
$$\{\lambda(2-a),\lambda(3+a)\} \subseteq \{1,6,7\}\ \ \mbox{and}\ \ \{\lambda(3-a),\lambda(4+a)\} \subseteq \{2,6,7\}.$$
Assuming that none of cases 1 or 2 occurs, we have two subcases to consider:

\begin{enumerate}
\item $\lambda(6)=2$.\\
Considering the clique $C(4)$ in $G(\{1,a,a+1\}^2)$, we get
$$\{\lambda(4-a),\lambda(5+a)\} \subseteq \{3,\lambda(2-a),\lambda(3+a)\}\setminus\{1,2\}=\{3,6,7\}.$$
If $\{\lambda(4-a),\lambda(5+a)\}=\{3,6\}$, then 
$$\{\lambda(3-a),\lambda(4+a)\}=\{2,7\},$$
$$\{\lambda(2-a),\lambda(3+a)\}=\{1,6\},$$ 
$$\{\lambda(1-a),\lambda(2+a)\}=\{5,7\}$$
and 
$$\{\lambda(-a),\lambda(1+a)\}=\{4,6\}.$$
If $\lambda(-a)=6$ then $\lambda(2-a)=1$ and thus $\lambda(4-a)=\lambda(-a)=6$ which corresponds to subcase 2.
If $\lambda(1+a)=6$ then $\lambda(3+a)=1$ and thus $\lambda(5+a)=\lambda(1+a)=6$ which again corresponds to subcase 2.

The case $\{\lambda(4-a),\lambda(5+a)\}=\{3,7\}$ is similar and leads to the same conclusion.

Finally, if $\{\lambda(4-a),\lambda(5+a)\}=\{6,7\}$ then $\lambda(3-a)=\lambda(4+a)=1$, a contradiction
since $d_{G(\{1,a,a+1\})}(4-a,5+a)=2$.

\item $\lambda(6)=6$.\\
Considering the clique $C(4)$ in $G(\{1,a,a+1\}^2)$, we get
$$\{\lambda(4-a),\lambda(5+a)\} \subseteq \{3,\lambda(2-a),\lambda(3+a)\}\setminus\{1,6\}=\{3,7\}.$$
This implies
$$\{\lambda(3-a),\lambda(4+a)\}=\{2,6\},$$
$$\{\lambda(2-a),\lambda(3+a)\}=\{1,7\},$$ 
$$\{\lambda(1-a),\lambda(2+a)\}=\{5,6\}$$
and 
$$\{\lambda(-a),\lambda(1+a)\}=\{4,7\}.$$
If $\lambda(-a)=7$ then $\lambda(2-a)=1$ and thus $\lambda(4-a)=\lambda(-a)=7$ which corresponds to subcase 2.
If $\lambda(1+a)=7$ then $\lambda(3+a)=1$ and thus $\lambda(5+a)=\lambda(1+a)=7$ which again corresponds to subcase 2.

\end{enumerate}

\item Vertices $0,1,2,3,4,5,6$ are colored with the colors $1,2,3,4,5,6$ and 1, respectively.\\
Again considering the cliques $C(1)$, $C(2)$ and $C(3)$ in $G(\{1,a,a+1\}^2)$, we get
$$\{\lambda(-a),\lambda(1-a),\lambda(1+a),\lambda(2+a)\}=\{4,5,6,7\},$$
$$\{\lambda(2-a),\lambda(3+a)\} \subseteq \{1,\lambda(-a),\lambda(1+a)\}\setminus\{4,5\},$$
and thus
$$\{\lambda(3-a),\lambda(4+a)\} \subseteq \{2,\lambda(1-a),\lambda(2+a)\}\setminus\{4,5,6\}=\{2,7\}.$$
This implies
$$\{\lambda(2-a),\lambda(3+a)\}=\{1,6\},$$ 
$$\{\lambda(1-a),\lambda(2+a)\}=\{5,7\}$$
and 
$$\{\lambda(-a),\lambda(1+a)\}=\{4,6\}.$$
Therefore, 
$$\{\lambda(4-a),\lambda(5+a)\} \subseteq \{3,\lambda(2-a),\lambda(3+a)\}\setminus\{1,6\}=\{3\},$$
a contradiction since $d_{G(\{1,a,a+1\})}(4-a,5+a)=2$.
\end{enumerate}

Therefore, every 2-distance $7$-coloring $\lambda$ of $G(\{1,a,a+1\})$ is necessarily $7$-periodic,
and thus $\chi_2(G(\{1,a,a+1\}))=7$ if and only if $7$ do not divide any element of 
$\{1,2,a-1,a,a+1,a+2,2a,2a+1,2a+2\}$.
This is clearly the case if and only if $a\equiv 2\pmod 7$ or $a\equiv 4\pmod 7$. 

\end{proof}

The following result provides an upper bound on $\chi_2(G(\{1,a,a+1\}))$ for any value of $a$.

\begin{theorem}
For every integer $a$, $a\ge 3$, $\chi_2(G(\{1,a,a+1\}))\le 9=\Delta(G(\{1,a,a+1\}))+3$.
\label{th:upper-bound-chi_2-dist-1aa+1}
\end{theorem}

\begin{proof}
We consider three cases, according to the value of $(a\mod 3)$:
\begin{enumerate}
\item $a=3k$, $k\ge 1$.\\
Let $\lambda$ be the $3a$-periodic mapping defined by the pattern
$$(123)^k(456)^k(789)^k.$$
If $\lambda(x)=\lambda(y)=c$, $1\le c\le 9$, then 
$$d(x,y)\in\{3q,\ 0\le q\le k-1\}\ \cup\ \{3ap+3q,\ p\ge 1,\ 1-k\le q\le k-1\}.$$
Therefore, $d(x,y)\not\in\{1,2,a-1,a,a+1,a+2,2a,2a+1,2a+2\}$, and thus $\lambda$ is a 2-distance 9-coloring of $G(\{1,a,a+1\})$. 

\item $a=3k+1$, $k\ge 1$.\\
Let $\lambda$ be the $(3a+2)$-periodic mapping defined by the pattern
$$(123)^k(456)^k7123(789)^{k-1}4568.$$
If $\lambda(x)=\lambda(y)=c$, $1\le c\le 6$, then 
$$\begin{array}{rl}
d(x,y)\in & \{3q,\ 0\le q\le k-1\}\\
 & \cup\ \{3q+2a-1,\ 1-k\le q\le 0\}\\
 & \cup\ \{(3a+2)p+2a-1,\ p>0\}\\
 & \cup\ \{(3a+2)p-2a+1,\ p>0\}\\
 & \cup\ \{(3a+2)p+3q,\ p>0,\ 1-k\le q<0\}\\
 & \cup\ \{(3a+2)p+3q+2a-1,\ p>0,\ 1-k\le q<0\}\\
 & \cup\ \{(3a+2)p+3q,\ p>0,\ 0<q\le k-1\}\\
 & \cup\ \{(3a+2)p+3q-2a+1,\ p>0,\ 0<q\le k-1\}.
\end{array}$$
If $\lambda(x)=\lambda(y)=7$, then 
$$\begin{array}{rl}
d(x,y)\in & \{3q,\ 0\le q\le k-2\}\\
 & \cup\ \{3q+4,\ 0\le q\le k-2\}\\
 & \cup\ \{(3a+2)p+3q-4,\ p>0,\ 2-k\le q\le 0\}\\
 & \cup\ \{(3a+2)p+3q+4,\ p>0,\ 0\le q\le k-2\}\\
 & \cup\ \{(3a+2)p+3q,\ p>0,\ 2-k\le q\le k-2\}.
\end{array}$$
If $\lambda(x)=\lambda(y)=8$, then 
$$\begin{array}{rl}
d(x,y)\in & \{3q,\ 0\le q\le k-2\}\\
 & \cup\ \{3q+a-2,\ 2-k\le q\le 0\}\\
 & \cup\ \{(3a+2)p+a-2,\ p>0\}\\
 & \cup\ \{(3a+2)p-a+2,\ p>0\}\\
 & \cup\ \{(3a+2)p+3q,\ p>0,\ 2-k\le q<0\}\\
 & \cup\ \{(3a+2)p+3q+a-2,\ p>0,\ 2-k\le q<0\}\\
 & \cup\ \{(3a+2)p+3q,\ p>0,\ 0<q\le k-2\}\\
 & \cup\ \{(3a+2)p+3q-a+2,\ p>0,\ 0<q\le k-2\}.
\end{array}$$
If $\lambda(x)=\lambda(y)=9$, then 
$$d(x,y)\in\{3q,\ 0\le q\le k-2\}\ \cup\ \{(3a+2)p+3q,\ p\ge 1,\ 2-k\le q\le k-2\}.$$
Therefore, in all these cases, $d(x,y)\not\in\{1,2,a-1,a,a+1,a+2,2a,2a+1,2a+2\}$, and thus $\lambda$ is a 2-distance 9-coloring of $G(\{1,a,a+1\})$. 

\item $a=3k+2$, $k\ge 1$.\\
Let $\lambda$ be the $(3a+1)$-periodic mapping defined by the pattern
$$(123)^{k+1}(456)^{k+1}(789)^k7.$$
If $\lambda(x)=\lambda(y)=c$, $1\le c\le 7$, then 
$$d(x,y)\in\{3q,\ 0\le q\le k\}\ \cup\ \{(3a+1)p+3q,\ p\ge 1,\ -k\le q\le k\}.$$
If $\lambda(x)=\lambda(y)=c$, $8\le c\le 9$, then 
$$d(x,y)\in\{3q,\ 0\le q\le k-1\}\ \cup\ \{(3a+1)p+3q,\ p\ge 1,\ 1-k\le q\le k-1\}.$$
Therefore, in both cases, $d(x,y)\not\in\{1,2,a-1,a,a+1,a+2,2a,2a+1,2a+2\}$, and thus $\lambda$ is a 2-distance 9-coloring of $G(\{1,a,a+1\})$. 
\end{enumerate}
This concludes the proof. 
\end{proof}

From Theorems~\ref{th:chi_2-dist-1aa+1} and~\ref{th:upper-bound-chi_2-dist-1aa+1},
we thus get:

\begin{corollary}
For every integer $a$, $a\ge 3$, $a\not\equiv 2,4\pmod 7$, 
$$8\le \chi_2(G(\{1,a,a+1\}))\le 9.$$
\label{cor:chi_2-1aa+1-encadrement}
\end{corollary}


\section{The case $D=\{1,\dots,m,a\}$, $2\le m<a$}
\label{sec:1toma}

We study in this section the 2-distance 
chromatic number of distance graphs $G(D)$
with $D=\{1,\dots,m,a\}$, $2\le m<a$ 
(note that the case $a=m+1$ is already solved by Proposition~\ref{prop:chi2-1k}).

When $D=\{1,\dots,m,a\}$, we have $\Delta(G(D))=2m+2$ and
$$D^2=\{1,2,\dots,2m\}\ \cup\ \{a-m,a-m+1,\dots,a+m\}\ \cup\ \{2a\}.$$

We first prove the following:

\begin{theorem}
For all integers $m$ and $a$, $2\le m<a$, 
$$\chi_2(G(\{1,\dots,m,a\}))=2m+3=\Delta(G(\{1,\dots,m,a\}))+1$$
 if and only if 
$a\equiv m+1\pmod{2m+3}$ or $a\equiv m+2\pmod{2m+3}$.
\label{th:chi_2-dist-1ma}
\end{theorem}

\begin{proof}
Since $\{1,\dots,m,a\}^2=\{1,\dots,2m\}\cup\{a-m,a-m+1,\dots,a+m\}\cup \{2a\}$,
$d\not\equiv 0\pmod{2m+3}$ for every $d\in \{1,\dots,m,a\}^2$ whenever $a\equiv m+1\pmod{2m+3}$ or $a\equiv m+2\pmod{2m+3}$,
 and thus, by Proposition~\ref{prop:k-upper-bound}
and Observation~\ref{obs:lower-bound},
$\chi_2(G(\{1,\dots,m,a\}))=2m+3$.

We now claim that every 2-distance $(2m+3)$-coloring $\lambda$ of $G(\{1,\dots,m,a\})$ is necessarily $(2m+3)$-periodic,
that is $\lambda(x+2m+3)=\lambda(x)$ for every $x\in\ZZZ$.
To show that, it suffices to prove that any $2m+3$ consecutive vertices $x,\dots,x+2m+2$ must be assigned distinct colors.
Assume to the contrary that this is not the case and, without loss of generality, let $x=0$.
Since vertices $0,1,\dots,2m$ necessarily get distinct colors, we only have two cases to consider:

\begin{enumerate}

\item Vertices $0,1,\dots,2m+1$ are colored with the colors $1,2,\dots,2m+1$ and 1, respectively.\\
Note that vertices $m-a$ and $m+a$ are both adjacent to all vertices $0,1,\dots,2m$. 
Hence, 
$$\{\lambda(m-a),\lambda(m+a)\}=\{2m+2,2m+3\},$$
which implies 
$$\{\lambda(m+1-a),\lambda(m+1+a)\}=\{2m+2,2m+3\}$$
(more precisely, $\lambda(m+1-a)=4m+5-\lambda(m-a)$ and $\lambda(m+1+a)=4m+5-\lambda(m+a)$).
This implies $\lambda(m+2-a)=\lambda(m+2+a)=2$, a contradiction since $d(m+2-a,m+2+a)=2$.

\item Vertices $0,1,\dots,2m+1,2m+2$ are colored with the colors $1,2,\dots,2m+1,2m+2$ and 1, respectively.\\
As in the previous case we have 
$$\{\lambda(m-a),\lambda(m+a)\}=\{2m+2,2m+3\},$$
which implies  
$$\{\lambda(m+1-a),\lambda(m+1+a)\}=\{1,2m+3\}.$$
We thus get $\lambda(m+2-a)=\lambda(m+2+a)=2$, again a contradiction.

\end{enumerate}

Therefore, every 2-distance $(2m+3)$-coloring $\lambda$ of $G(\{1,\dots,m,a\})$ is necessarily $(2m+3)$-periodic,
and thus $\chi_2(G(\{1,\dots,m,a\}))=2m+3$ if and only if $2m+3$ do not divide any element of 
$\{1,2,\dots,2m\}\ \cup\ \{a-m,a-m+1,\dots,a+m\}\ \cup\ \{2a\}$.
This is clearly the case if and only if $a\equiv m+1\pmod{2m+3}$ or $a\equiv m+2\pmod{2m+3}$. 

\end{proof}


For other values of $a$, we propose the following general upper bound on

\begin{theorem}
For all integers $m$ and $a$, $2\le m<a$, 
$$\chi_2(G(\{1,\dots,m,a\}))\le 4m+2=2\Delta(G(\{1,\dots,m,a\}))-2.$$
\label{th:upper-bound-chi_2-dist-1ma}
\end{theorem}

\begin{proof}
Let $a=(2m+1)k+r$, $0\le r < 2m+1$.
We consider four cases, depending on the value of $r$.
In each case, we will provide a periodic 2-distance $(4m+2)$-coloring of 
the distance graph $G(\{1,\dots,m,a\})$.
\begin{enumerate}

\item $r < m$.\\ 
Let $\lambda$ be the $(2a-r-m)$-periodic mapping defined by the pattern
$$[12\ldots (2m+1)]^k[(2m+2)(2m+1)\ldots (4m+2)]^{k-1}(2m+2)(2m+3)\ldots (3m+r+2).$$
If $\lambda(x)=\lambda(y)=c$, $1\le c\le 3m+r+2$, then
$$\begin{array}{rl}
d(x,y)\in & \{q(2m+1),\ 0\le q\le k-1\}\\
 & \cup\ \{p(2a-r-m)+q(2m+1),\ p\ge 1,\ 1-k\le q\le k-1\}.
\end{array}$$
If $\lambda(x)=\lambda(y)=c$, $3m+r+3\le c\le 4m+2$, then 
$$\begin{array}{rl}
d(x,y)\in & \{q(2m+1),\ 0\le q\le k-2\}\\
 & \cup\ \{p(2a-r-m)+q(2m+1),\ p\ge 1,\ 2-k\le q\le k-2\}.
\end{array}$$
Therefore, in both cases, $d(x,y)\not\in\{1,2,\dots,2m\}\cup\{a-m,a-m+1,\dots,a+m\}\cup\{2a\}$, 
and thus $\lambda$ is a 2-distance $(4m+2)$-coloring of $G(\{1,\dots,m,a\})$.

\item $r=m$.\\
Let $\lambda$ be the $(2a-2m)$-periodic mapping defined by the pattern
$$[12\ldots (2m+1)]^k[(2m+2)(2m+1)\ldots (4m+2)]^k.$$
If $\lambda(x)=\lambda(y)=c$, $1\le c\le 4m+2$, then
$$\begin{array}{rl}
d(x,y)\in & \{q(2m+1),\ 0\le q\le k-1\}\\
 & \cup\ \{p(2a-2m)+q(2m+1),\ p\ge 1,\ 1-k\le q\le k-1\}.
\end{array}$$
Therefore, $d(x,y)\not\in\{1,2,\dots,2m\}\cup\{a-m,a-m+1,\dots,a+m\}\cup\{2a\}$, 
and thus $\lambda$ is a 2-distance $(4m+2)$-coloring of $G(\{1,\dots,m,a\})$.

\item $r=m+1$.\\
Let $\lambda$ be the $(2a+1)$-periodic mapping defined by the pattern
$$[12\ldots (2m+1)]^{k+1}[(2m+2)(2m+1)\ldots (4m+2)]^k(2m+2)(2m+3).$$
If $\lambda(x)=\lambda(y)=c$, $1\le c\le 2m+3$, then
$$\begin{array}{rl}
d(x,y)\in & \{q(2m+1),\ 0\le q\le k\}\\
 & \cup\ \{p(2a+1)+q(2m+1),\ p\ge 1,\ -k\le q\le k\}.
\end{array}$$
If $\lambda(x)=\lambda(y)=c$, $2m+4\le c\le 4m+2$, then 
$$\begin{array}{rl}
d(x,y)\in & \{q(2m+1),\ 0\le q\le k-1\}\\
 & \cup\ \{p(2a+1)+q(2m+1),\ p\ge 1,\ 1-k\le q\le k-1\}.
\end{array}$$
Therefore, in both cases, $d(x,y)\not\in\{1,2,\dots,2m\}\cup\{a-m,a-m+1,\dots,a+m\}\cup\{2a\}$, 
and thus $\lambda$ is a 2-distance $(4m+2)$-coloring of $G(\{1,\dots,m,a\})$.

\item $m+2 \le r < 2m+1$.\\
Let $\lambda$ be the $(2a-r+m+1)$-periodic mapping defined by the pattern
$$[12\ldots (2m+1)]^{k+1}[(2m+2)(2m+1)\ldots (4m+2)]^{k}(2m+2)(2m+3)\ldots (m+r+1).$$
If $\lambda(x)=\lambda(y)=c$, $1\le c\le m+r+1$, then
$$\begin{array}{rl}
d(x,y)\in & \{q(2m+1),\ 0\le q\le k\}\\
 & \cup\ \{p(2a-r+m+1)+q(2m+1),\ p\ge 1,\ -k\le q\le k\}.
\end{array}$$
If $\lambda(x)=\lambda(y)=c$, $m+r+2\le c\le 4m+2$, then 
$$\begin{array}{rl}
d(x,y)\in & \{q(2m+1),\ 0\le q\le k-1\}\\
 & \cup\ \{p(2a-r+m+1)+q(2m+1),\ p\ge 1,\ 1-k\le q\le k-1\}.
\end{array}$$
Therefore, in both cases, $d(x,y)\not\in\{1,2,\dots,2m\}\cup\{a-m,a-m+1,\dots,a+m\}\cup\{2a\}$, 
and thus $\lambda$ is a 2-distance $(4m+2)$-coloring of $G(\{1,\dots,m,a\})$. 

\end{enumerate}
This concludes the proof.
\end{proof}

From Theorems~\ref{th:chi_2-dist-1ma} and~\ref{th:upper-bound-chi_2-dist-1ma},
we thus get:

\begin{corollary}
For all integers $m$ and $a$, $2\le m<a$, $a\not\equiv m+1,m+2\pmod{2m+3}$, 
$$2m+4\le \chi_2(G(\{1,\ldots,m,a\}))\le 4m+2.$$
\label{cor:chi_2-dist-1ma-encadrement}
\end{corollary}

\section{Discussion}
\label{sec:discussion}

In this paper, we studied 2-distance colorings of several types of distance graphs.
In each case, we characterized those distance graphs that admit an optimal 2-distance
coloring, that is distance graphs $G(D)$ with $\chi_2(G(D))=\Delta(G(D))+1$.
We also provided general upper bounds for the 2-distance chromatic number of the
considered graphs.

We leave as open problems the question of completely determining the 2-distance
chromatic number of distance graphs $G(D)$ when 
$D=\{1,a,a+1\}$, $a\ge 3$,
or $D=\{1,\dots,m,a\}$, $2\le m<a$.

Considering other types of sets $D$ would certainly be also an interesting direction for future research.

\bigskip

\noindent{\bf Acknowledgement.}
Most of this work has been done while the first author was visiting LaBRI, thanks to a grant from 
University of Sciences and Technology Houari Boumediene (USTHB).
The second author was partially supported by the Cluster of excellence CPU, from the
Investments for the future Programme IdEx Bordeaux (ANR-10-IDEX-03-02).



\begin{thebibliography}{99}

\bibitem{AH03} G. Agnarsson and M.M. Halld\'orsson.
Coloring powers of planar graphs.
{\em SIAM J. Discrete Math.} 16(4):651--662 (2003).


\bibitem{BLP14a} M. Bonamy, B. L\'ev\^eque and A. Pinlou. 
2-distance coloring of sparse graphs.
{\em J. Graph Theory} 77(3):190--218 (2014).

\bibitem{BLP14b} M. Bonamy, B. L\'ev\^eque and A. Pinlou. 
Graphs with maximum degree $\Delta\ge 17$ and maximum average degree less than 3 are list 2-distance 
$(\Delta+2)$-colorable. 
{\em Discrete Math.} 317:19--32 (2014).

\bibitem{BI09} O.V. Borodin and A.O. Ivanova. 
2-distance $(\Delta + 2)$-coloring of planar graphs with
girth six and $\Delta\ge 18$. 
{\em Discrete Math.} 309:6496--6502 (2009).


\bibitem{Ca11} T. Calamoneri. 
The $L(h,k)$-labelling problem: An updated survey and annotated bibliography.
{\em The Computer Journal} 54(8):1344-1371 (2011).

\bibitem{DKNS08} Z. Dvo\v r\`ak, D. Kr\`al, P. Nejedl\`y, and R. \v Skrekovski. 
Coloring squares of planar graphs with girth six. 
{\em European J. Combin.} 29(4):838--849 (2008).

%
%
%


\bibitem{KK69} F. Kramer and H. Kramer.
Un probl\`eme de coloration des sommets d'un graphe. 
{\em C. R. Acad. Sci. Paris A} 268(7):46--48  (1969).

\bibitem{KK08} F. Kramer and H. Kramer. 
A survey on the distance-coloring of graphs.
{\em Discrete Math.} 308:422--426  (2008).

%

\bibitem{LW06} K.W. Lih and W.F. Wang.
Coloring the square of an outerplanar graph.
{\em Taiwan. J. Math.} 10(4):1015--1023 (2006).

\bibitem{Li08} D. D.-F. Liu.
From rainbow to the lovely runner: A survey on coloring
parameters of distance graphs.
{\em Taiwanese J. Math.} 12(4):851-871 (2008).



%
%

%

%
%
%

\bibitem{W90} H. Walther.
\"Uber eine spezielle Klasse unendlicher Graphen.
In K.~Wagner and R.~Bodendiek, eds, {\em Graphentheorie}, vol. 2, pp. 268--295 (1990),
Bibl. Inst., Mannheim.


\bibitem{WL03}	W.F. Wang and K.W. Lih. 
Labeling planar graphs with conditions on girth and distance two. 
{\em SIAM J. Discrete Math.} 17(2):264-275 (2003).

\bibitem{W77} G. Wegner. 
Graphs with given diameter and a colouring problem. 
Technical Report, University of Dortmund (1977).


\bibitem{Z02} X. Zhu.
Circular chromatic number of distance graphs with distance sets of cardinality three.
{\em J. Graph Theory} 41(3):195--207 (2002).

\end{thebibliography}
\end{document}